\documentclass[pdftex,11pt,svgnames]{article}

\usepackage[utf8]{inputenc}            
\usepackage{amsmath}
\usepackage{amssymb}  
\usepackage{amsbsy} 
\usepackage{listings}
\usepackage{graphics}

\usepackage[multiple]{footmisc}

\usepackage{enumerate}
\usepackage[shortlabels]{enumitem}

\IfFileExists{dsfont.sty}{
	\usepackage{dsfont}
	\usepackage[usenames,dvipsnames]{color} 
	
}{
	
}

\usepackage[numbers]{natbib}

\usepackage{pdfpages}
\usepackage{hyperref} 
\hypersetup{
    unicode=false,          
    pdftoolbar=true,        
    pdfmenubar=true,        
    pdffitwindow=false,     
    pdfstartview={FitH},    
    pdftitle={Tax- and expense-modified risk-minimization for insurance payment processes},    
    pdfauthor={Kristian Buchardt, Christian Furrer, and Thomas Møller},     
    pdfsubject={Life Insurance Mathematics},   
    pdfcreator={ },   
    pdfproducer={ }, 
    pdfkeywords={ } { } { }, 
    pdfnewwindow=true,      
    colorlinks=true,       
    linkcolor=Black, 
    citecolor=Black,        
    filecolor=Black,      
    urlcolor=Black 
}

\usepackage[font=small,format=plain,labelfont=bf,up,textfont=it,up]{caption}

\usepackage[a4paper]{geometry}  
\usepackage[english]{babel}

\usepackage[autolanguage,np]{numprint}

\usepackage{doi}

\usepackage{fancyhdr}
\usepackage{amsthm}
\theoremstyle{plain}
\newtheorem{Theorem}{Theorem}
\newtheorem{theorem}[Theorem]{Theorem}

\newtheorem{lemma}[Theorem]{Lemma}

\newtheorem{proposition}[Theorem]{Proposition}

\newtheorem{definition}[Theorem]{Definition}

\theoremstyle{definition}

\theoremstyle{remark}

\newtheorem{remark}[Theorem]{Remark}

\numberwithin{Theorem}{section}

\numberwithin{equation}{section}

\newcommand{\md}{{\, \mathrm{d} }}
\newcommand{\pderiv}[2]{\frac{\partial {#1}}{\partial {#2}}}

\DeclareMathOperator{\E}{E}

\newcommand{\1}[1]{1_{\{ #1\}}}

\makeatletter

\newcommand{\Rmnum}[1]{\expandafter\@slowromancap\romannumeral #1@}
\makeatother

\usepackage{wasysym} 

\usepackage{tikz}
\usetikzlibrary{arrows,positioning,calc} 
\tikzset{
    >=stealth',
    punkt/.style={
           rectangle,
           rounded corners,
           draw=black, thick,
           text width=5em,
           minimum height=2em,
           text centered},
    punktl/.style={
           re
           tangle,
           rounded corners,
           draw=black, thick,
           
           text width=7em,
           minimum height=2em,
           text centered},
    pil/.style={
           ->,
           shorten <=4pt,
           shorten >=4pt,},
    pildotted/.style={
           ->,
           shorten <=4pt,
           shorten >=4pt,
  dotted,
  }
}
\usepackage{subfig}
\usepackage{todonotes}

\usepackage{authblk}

\makeatletter
\let\@fnsymbol\@arabic
\makeatother

\title{Tax- and expense-modified risk-minimization for insurance payment processes}
\author[1]{Kristian Buchardt}
\author[1,2,$\star$]{Christian Furrer}
\author[2,3]{Thomas M\o ller}
\affil[1]{\footnotesize PFA Pension, Sundkrogsgade 4, DK-2100 Copenhagen \O, Denmark.}
\affil[2]{\footnotesize Department of Mathematical Sciences, University of Copenhagen, Universitetsparken 5, DK-2100 Copenhagen \O, Denmark.}
\affil[3]{\footnotesize AP Pension, \O stbanegade 135, DK-2100 Copenhagen \O, Denmark.}
\affil[$\star$]{\footnotesize Corresponding author. E-mail: \href{mailto:furrer@math.ku.dk}{furrer@math.ku.dk}.}

\date{}

\begin{document}
\maketitle

\addtocounter{footnote}{4} 

\begin{center}
{\sc Abstract}
\end{center}
{\small

We study the problem of determining risk-minimizing investment strategies for insurance payment processes in the presence of taxes and expenses. We consider the situation where taxes and expenses are paid continuously and symmetrically and introduce the concept of tax- and expense-modified risk-minimization. Risk-minimizing strategies in the presence of taxes and expenses are derived and linked to Galtchouk-Kunita-Watanabe decompositions associated with modified versions of the original payment processes. Furthermore, we show equivalence to an alternative approach involving an artificial market consisting of after-tax and after-expense assets, and we establish -- in a certain sense -- consistency with classic risk-minimization. Finally, a case study involving classic multi-state life insurance payments in combination with a bond market exemplifies the results. 

\vspace{5mm}

\textbf{Keywords:} Quadratic hedging; Incomplete market; Market consistent valuation; Intrinsic value process; Galtchouk-Kunita-Watanabe decomposition

\vspace{5mm}

\textbf{2010 Mathematics Subject Classification:} 62P05; 91G99; 91B30; 60G44

\textbf{JEL Classification:} G10; G22 

}
 
\section{Introduction}\label{sec:intro}
According to a recent OECD report on taxation of funded private pension plans across different countries \cite{oecd2015}, taxes on pension fund returns are widespread. Consequently, market consistent valuation of such insurance liabilities requires one to take into account the associated future tax payments, which are closely related to the investment strategy. Similarly, future expenses associated with the management of the insurance contract and investment strategy should be included in considerations about hedging and valuation. The necessity to take taxes and expenses into account is also reflected in the Solvency II regulation, see \cite{eiopa2009} Article 77--78 and \cite{eiopa2015} Article 28, and the forthcoming IFRS17 regulation, see \cite{ifrs17} Paragraph 34 and Paragraph B65(j).
It is our impression that a unified theory for market consistent valuation in the presence of taxes and expenses is yet to be developed, and accordingly, it is common among practitioners to take taxes and expenses into account via certain ad hoc adjustments of the forward interest rate curve, confer with~\cite{buchardtmoeller2018}.

In this paper, we consider quadratic hedging of insurance payment processes in the presence of taxes and expenses. We allow for idealized taxes and expenses which depend on the investment strategy and develop the concept of tax- and expense-modified risk-minimization. The taxes are defined as a fraction of the returns from the investment strategy, and the expenses are defined as a fraction of the value of the investment strategy, thus both are symmetrical and continuously paid.
The primary idea is to introduce a tax- and expense-modified version of the so-called cost process and then minimize at any time the associated risk process, which is defined as the conditional expected value of the squared future tax- and expense-modified costs given the information currently available.
As our main result, we show the existence and uniqueness of an optimal strategy and relate it to the Galtchouck-Kunita-Watanabe decomposition of the intrinsic value process associated with a tax- and expense-modified payment process.

Given a payment process in an incomplete market, it is well studied how to apply the quadratic hedging criterion of risk-minimization to find an optimal investment strategy and price the contract. The criterion of risk-minimization was originally proposed by~\cite{fs1986} and was extended to insurance payment processes in~\cite{thmoller01}; for an overview, see~\cite{schweizerguided}. Tax- and expense-modified risk-minimization differs from classic risk-minimization, in essence because a tax- and expense-modified savings account is used as numeraire. We show that tax- and expense-modified risk-minimization is consistent with classic risk-minimization in the sense that a subsequent application of classic risk-minimization confirms the investment strategy, thus not reducing the risk further.

In addition to the tax- and expense-modified risk-minimization approach, we also solve the problem by creating an artificial after-tax and after-expense market: The assets are constructed such that the returns are after payment of taxes and expenses. In this market, we are able to apply classic risk-minimization and thereby find an optimal investment strategy, which is essentially identical to the investment strategy from tax- and expense-modified risk-minimization. This is a consequence of the cost processes in the two approaches being in a certain sense identical.

Taxes on investment returns and expenses associated with the management of the insurance contract and investment strategy can be viewed as negative dividends. In that sense, the concept of tax- and expense-modified risk-minimization corresponds to a kind of risk-minimization in the presence of negative dividends. While the extension of risk-minimization to include transaction costs is studied in depth in the literature, see e.g.\ \cite{lambertonetal} and \cite{guasoni2002}, there does not seem to be a similar treatment in the literature of the case of dividends; an exception being \cite{battauz2003} on quadratic hedging in the presence of discrete stochastic dividends.

In \cite{buchardtmoeller2018}, valuation of insurance payment processes in the presence of symmetric and continuously paid taxes and expenses is studied in complete markets. By identifying and explicitly constructing the inherent tax and expense payment processes and adding these to the existing insurance payment process, \cite{buchardtmoeller2018} were able to derive replicating strategies and determine the market value of the combined liability. In particular, by disregarding systematic insurance risk and implicitly dealing with unsystematic insurance risk via diversification, \cite{buchardtmoeller2018} were able to argue that current actuarial practice is prudent. In this paper, we essentially extend the results of~\cite{buchardtmoeller2018} to allow for unsystematic as well as systematic insurance risk via an incomplete market approach.

Having determined the value of the combined liability, as well as the associated risk-minimizing investment strategy, it is interesting, as was done for complete markets in \cite{buchardtmoeller2018}, to study the decomposition into benefit, tax, and expense parts. This is, however, together with extensions of the tax- and expense setup to include asymmetrical and discrete payments, postponed to future research.

The paper is structured in the following way. In Section~\ref{sec:rm}, we give a brief review of the main results on risk-minimization for insurance payment processes. In Section~\ref{sec:tmrm}, we study risk-minimization in the presence of symmetrical and continuously paid taxes and expenses. Section~\ref{sec:ex} concludes with a case study, where we consider classic multi-state life insurance payments in a bond market for a constant tax rate and expenses depending on the current state of the insurance contract(s).

\section{Risk-minimization for insurance payment processes}\label{sec:rm}

In this section, we give a brief review of the main results on risk-minimization for insurance payment processes from~\cite{thmoller01}, see also~\cite{schweizerguided}, before introducing taxes and expenses in the next section. Regarding the technical details and the necessary regularity conditions, we generally refer to \cite{thmoller01} and the references therein.

Consider an arbitrage-free financial market consisting of $d+1$ traded assets with price processes ${S}_0$ and $({S}_1,\ldots, {S}_d)$ defined on a probability space $(\Omega,\mathcal{F},P)$ equipped with a filtration $\mathbb{F}=(\mathcal{F}(t))_{t\in [0,T]}$ satisfying the usual conditions with $\mathcal{F}(0)$ trivial. Here $T>0$ is a fixed finite time horizon.  We assume that ${S}_0$ is the savings account and that it is on the form 
\begin{align*}
S_0(t) =\exp\!\left(\int_0^t r(u) \md u\right)\!,
\end{align*}
where ${r}$ is the so-called short rate process.

All quantities are modeled directly under an equivalent martingale measure $Q$, such that the discounted price processes ${S}^*_j = {S}_j/{S}_0$ are $Q$-martingales. In general, results hold almost surely w.r.t.\ $Q$. We discuss the choice of equivalent martingale measure in the last paragraph of the present section.

We study an undiscounted \textit{insurance payment process}, which is a stochastic process ${A}$ describing the accumulated benefits less premiums associated with some insurance contract(s).

Let ${S}^* = \left( {S}^*_1, \ldots, {S}^*_d\right)$. Following~\cite{schweizer1994, schweizer2008}, there exists a bounded, strictly increasing, predictable process $B$, null at $0$, such that
\begin{align*}
\langle S^*_i, S^*_j \rangle \ll B 
\end{align*}
with $\langle \cdot \rangle$ denoting the predictable variation. Define matrix-valued predictable process $\sigma_S$ by
\begin{align}\label{eq:sigmaB}
\md \langle S^* \rangle = \sigma_S \md B.
\end{align}
Here each $\sigma_S(t)$ is a positive semidefinite symmetric $d \times d$-matrix. To ensure uniqueness of certain decompositions and optimal strategies in the sense that the amount invested in every asset is unique, we further assume that each $\sigma_S(t)$ is actually positive definite.

An \textit{investment strategy} ${h}$ is a $(d+1)$-dimensional process. Both the discounted price processes $S^*_j$, the insurance payment process $A$, the short rate process $r$, and the investment strategies $h$ satisfy certain regularity conditions. The undiscounted \textit{value process} ${V}$ associated with ${h}$ is defined by
 \begin{align}\label{eq:Vhtdef}
	V(h,t) = \sum_{j=0}^d h_j(t) S_j(t).
\end{align}
The value process measures the value of the investment strategy after the payments prescribed by ${A}$, i.e.\ $V(h,t)$ is the value of the investments after the payments $A(t)$ during $[0,t]$. We say that the investment strategy ${h}$ is \textit{$0$-admissible} if the value at time $T$ is $0$, i.e.\ if $V(h,T)=0$.

The undiscounted \textit{cost process} ${C}$ associated with ${h}$ is defined by
\begin{align}
	C(h,t) = V(h,t) - \sum_{j=0}^d \int_0^t h_j(u) \md S_j(u) + A(t). 
\end{align}
The value process measures the current value of the investment strategy ${h}$, and the cost process measures the accumulated costs associated with the investment strategy ${h}$ and the insurance payment process ${A}$. The accumulated costs at time $t$ are given by the current value of the investment portfolio, added past payments and reduced by realized trading gains.

Define the discounted value process by ${V}^* = {V}/{S}_0$, the discounted insurance payment process ${A}^*$ by ${A}^*(0) = 0$ and
\begin{align*}
	\md {A}^*(t) = {S}_0^{-1}(t) \md {A}(t),
\end{align*}
and the discounted cost process ${C}^*$ by ${C}^*(0)=0$ and
\begin{align}\label{eq:dC*htdef}
	\md {C}^*(t)  = {S}_0^{-1}(t) \md {C}(t).
\end{align}
It follows that
\begin{align}\label{eq:dC*ht}
	\md C^*(h,t ) = \md V^*(h,t ) - \sum_{j=1}^d h_j(t) \md S^*_j(t) + \md A^*(t).
\end{align}
The \textit{risk process} ${R}$ associated with ${h}$ and ${A}$ is defined by 
\begin{align}\label{eq:Rt}
	R(h,t) = \E^Q\!\left[\left.  \left(C^*(h,T)- C^*(h,t) \right)^2 \, \right| \mathcal{F}(t) \right]\!.
\end{align}
The process measures the quadratic risk under the measure $Q$ associated with the future costs $(C^*(h,T)-C^*(h,t))$ given the information currently available. An investment strategy $h$ is said to be \textit{risk-minimizing} for ${A}$ if it is $0$-admissible and minimizes the risk process at any point in time.

Following~\cite{fs1986} and \cite{thmoller01}, define the so-called \textit{intrinsic value process} $\mathcal{V}^{A^*}$ associated with $A^*$ by 
\begin{align}\label{eq:VA*T}
	\mathcal{V}^{A^*}(t) = \E^Q\!\left[ \left. A^*(T) \, \right| \mathcal{F}(t)\right]  = A^*(t) + 
	\E^Q\!\left[ \left. \int_t^T e^{-\int_0^s r(u) \md u} \md A(s) \, \right| \mathcal{F}(t)\right]\!.
\end{align}
There exists a unique decomposition for $\mathcal{V}^{A^*}$ on the form
\begin{align}\label{eq:VA*Tdec}
	\mathcal{V}^{A^*}(t) = \mathcal{V}^{A^*}(0) + \sum_{j=1}^d \int_0^t h^{A^*}_j(u) \md S^*_j(u) + L^{A^*}(t),
\end{align}
where $h^{A^*}_1,\ldots, h^{A^*}_d$ satisfy certain regularity conditions, and where $L^{A^*}$ is a zero-mean $Q$-martingale which is orthogonal to the discounted price processes $S^*$. The decomposition~(\ref{eq:VA*Tdec}) is also known as the Galtchouk-Kunita-Watanabe decomposition. 

Theorem 2.1 in~\cite{thmoller01} shows for the case $d=1$ (an extension to the multidimensional case is possible; the assumption of positive definiteness following~\eqref{eq:sigmaB} ensures uniqueness) that there exists a unique risk-minimizing investment strategy ${h}^*$ for ${A}$ given by $h^*_j = h^{A^*}_j$ for $j=1,\ldots, d$ and 
\begin{align}\label{eq:temp?}
	h^*_0(t)  = \mathcal{V}^{A^*}(t) - A^*(t) - \sum_{j=1}^d h_j^{A^*}(t) S^*_j(t).
\end{align}

Consequently, if one can explicitly write up the relevant Galtchouk-Kunita-Watanabe decomposition, this immediately yields an explicit risk-minimizing investment strategy for the insurance payment process.

The value process associated with the risk-minimizing investment strategy is
\begin{align}
	V(h^*,t)
	&=
	\mathcal{V}^{A^*}(t) - A^*(t), \nonumber \\
	&=
	\E^Q\!\left[ \left. \int_t^T e^{-\int_t^s r(u) \md u} \md A(s) \, \right| \mathcal{F}(t)\right]\!.
\end{align}
in particular, the value of the investments before any payments is
\begin{align}
	V(h^*,0-)
	:=
	V(h^*,0) + A(0)
	&=
	\mathcal{V}^{A^*}(0) \nonumber \\
	&=
	A(0)
	+
	\E^Q\!\left[ \int_0^T e^{-\int_0^s r(u) \md u} \md A(s) \right]
\end{align}
due to \eqref{eq:temp?}. 

If the discounted price processes $S^*$ are continuous, then the Föllmer-Schweizer decomposition under $P$ can be obtained as the Galtchouk-Kunita-Watanabe decomposition, computed under the so-called minimal martingale measure $\hat{Q}$, see e.g.~\cite{schweizer2008}. Consequently, for $Q=\hat{Q}$ risk-minimization is equivalent to local risk-minimization (for continuous discounted price processes). This makes the minimal martingale measure a natural candidate measure, but the concept of risk-minimization can in principle be used under any equivalent martingale measure.

\section{Risk-minimization in the presence of taxes and expenses}\label{sec:tmrm}

We extend the setting from Section~\ref{sec:rm} by including symmetrical taxes and expenses paid continuously. As in the previous section, we consider an insurance payment process ${A^{\text{b}}}$ describing the accumulated benefits less premiums associated with some insurance contract(s).  

We study two different approaches. First, we define after-tax and after-expense price processes directly from the underlying before-tax and before-expense price processes. These price processes are constructed exactly such that the return corresponds to the original return after taxes and expenses. In this setting with an artificial after-tax and after-expense market, we apply the criterion of risk-minimization directly . Second, we follow~\cite{buchardtmoeller2018} and construct explicitly the payment processes associated with taxes and expenses and introduce the concept of tax- and expense-modified risk-minimization. The two approaches are conceptually different but are, as we unveil, mathematically equivalent in a specific sense which we explain later. The first approach using after-tax and after-expense price processes is detailed in Subsection~\ref{subsec:after_tax}, while Subsection~\ref{subsec:modified_riskmin} deals with tax- and expense-modified risk-minimization.

Following the second approach, where we have employed tax- and expense-modified risk-minimization, the risk-minimizing investment strategy leads to specific tax payments and expense payments. We investigate the following question: if another investor assumes these payments, would it using classic risk-minimization as in Section~\ref{sec:rm} employ a different optimal investment strategy than the original investor, who used the criterion of tax- and expense-modified risk-minimization? Unsurprisingly, the answer turns out negative; it is not possible to reduce the risk further. The investigation is presented in Subsection~\ref{subsec:two-step}.


\subsection{Risk-minimization in the after-tax and after-expense market} \label{subsec:after_tax}

To model the taxes and expenses, we introduce a tax rate $\gamma$ and an expense rate $\delta$; both are adapted processes. We assume that $\gamma$ takes values in $[0,1)$, has limits from the left, and is bounded away from $1$, while $\delta$ is only assumed to be bounded (and measurable).  The taxes are paid continuously at rate $\gamma$ as a fraction of all returns (positive and negative) from the investment strategy, and the expenses are also paid continuously at rate $\delta$ but instead as a fraction of the value of the investment strategy.

The taxation and expense schemes introduced here are idealizations of real-life regimes. As an example, the Danish PAL-tax is a flat tax of 15.3 \% on pension fund returns, but it is paid on a yearly basis (non-continuously) and asymmetrically: negative yearly returns lead to future tax deductions rather than negative tax payments. By setting $\gamma\equiv0.153$ we obtain an idealization of the Danish taxation regime. In general, we consider our idealized approach a significant step towards understanding and handling a wide range of taxation and expense regimes.

Consider after-tax and after-expense price processes ${\check{S}}_j$ given by $\check{S}_j(0)={S}_j(0)$ and
\begin{align}\label{eq:S_check}
		\md \check{S}_j(t) = \check{S}_j(t-)\left( (1-\gamma(t-)) \frac{\md S_j(t)}{S_j(t-)}  -\delta(t) \md t \right)\!, 
\end{align}
for $j=0,1,\ldots, d$, where ${S}_j$ are the before-tax and before-expense price processes introduced in Section~\ref{sec:rm}. In the following, we assume that the fractions ${\check{S}}_j/{S}_j$ are well-defined and that there exists suitably regular (strong) solutions to \eqref{eq:S_check}. We interpret the after-tax and after-expense price processes as price processes of an artificial after-tax and after-expense market; this is based on the following observation: Rewriting \eqref{eq:S_check}, we see that
\begin{align*}
\frac{\md \check{S}_j(t)}{\check{S}_j(t-)} = (1-\gamma(t-)) \frac{\md S_j(t)}{S_j(t-)} - \delta(t) \md t,
\end{align*}
which shows that the relative returns of the after-tax and after-expense assets are affine transformations of the relative returns of the original before-tax and before-expense assets. The relative returns are scaled with a factor $(1-\gamma)$ and reduced by $\delta$. In other words, the returns~\eqref{eq:S_check} correspond to the returns obtained by an investor paying taxes and expenses according to the scheme described in the beginning of this subsection.

The after-tax and after-expense version of the savings account takes the form
\begin{align}\label{eq:S_0_check}
	\check{S}_0(t) 
	&=
	\exp\!\left(\int_0^t \left(\left(1-\gamma(u)\right)r(u) - \delta(u)\right) \md u\right) \nonumber \\
	&=
	e^{-\int_0^t \left(\gamma(u)r(u) + \delta(u)\right) \md u} S_0(t).
\end{align}
This can be interpreted as using an artificial after-tax and after-expense short rate $((1-\gamma)r-\delta)$ rather than the original short rate $r$. 

We now study the artificial after-tax and after-expense market $({\check{S}}_0,{\check{S}}_1,\ldots,{\check{S}}_d)$ within the setup of Section~\ref{sec:rm}. Thus, we use the after-tax and after-expense savings account ${\check{S}}_0$ as numeraire, and we search for a risk-minimizing investment strategy ${\check{h}}$ for an insurance payment process ${A^{\text{b}}}$ in the after-tax and after-expense market.

It can be shown that the discounted after-tax and after-expense price processes defined by ${\check{S}}_j^*= {\check{S}}_j/{\check{S}}_0$ have dynamics
\begin{align}\label{eq:S*_check}
		\md \check{S}^*_j(t)
		&=
		\check{S}^*_j(t-) \left(1-\gamma(t-)\right)\left( \frac{\md S_j(t)}{S_j(t-)}  -r(t) \md t \right) \nonumber \\
		&=
		\frac{\check{S}^*_j(t-)}{S_j(t-)} S_0(t) \left(1-\gamma(t-)\right) \left(\left(S_0(t)\right)^{-1} \md S_j(t) - r(t) S^*_j(t-) \md t\right) \nonumber \\
		&=
		\frac{\check{S}^*_j(t-)}{S^*_j(t-)} \left(1-\gamma(t-)\right) \!\md S^*_j(t).
\end{align} 
Note that $ {\check{S}}_0$ rather than ${S}_0$ is used as numeraire in the definition of the discounted after-tax  and after-expense price processes.

Since the discounted before-tax and before-expense price processes ${S}^*_j$ are $Q$-martingales, it follows from~\eqref{eq:S*_check} that the after-tax and after-expense price processes ${\check{S}}_j^*$ are $Q$-local martingales. In the following, we assume for simplicity that the after-tax and after-expense price processes actually are $Q$-martingales.

From \eqref{eq:S_0_check} we see that
\begin{align}\label{eq:dS*check}
		\md \check{S}^*_j(t)
		=
		\frac{\check{S}_j(t-)}{S_j(t-)} e^{\int_0^t \left(\gamma(u)r(u) + \delta(u)\right) \md u} \left(1-\gamma(t-)\right)\! \md S^*_j(t),
\end{align} 
which relates the dynamics of the discounted after-tax price and after-expense price processes to the dynamics of the discounted before-tax and before-expense price processes.
%

In this setting, the discounted insurance payment process $\check{A}^{\text{b},*}$ is given by $\check{A}^{\text{b},*}(0) = A^{\text{b}}(0)$ and
\begin{align}\label{eq:dAcheck*}
	\md \check{A}^{\text{b},*}(t) 
	=
	\check{S}_0^{-1}(t) \! \md A^{\text{b}}(t)
	=
	e^{\int_0^t \left(\gamma(u)r(u) + \delta(u)\right) \md u} S_0^{-1}(t) \md A^{\text{b}}(t),
\end{align}
the undiscounted and discounted value processes $\check{V}$ and ${\check{V}}^*$ associated with $h$ are
\begin{align*}
	\check{V}({h},t)
	&=
	\sum_{j=0}^{d}{h}(t) \check{S}_j(t), \\
	\check{V}^*({h},t)
	&=
	\sum_{j=0}^{d}{h}(t) \check{S}_j^*(t),
\end{align*}
and the undiscounted and discounted cost processes $\check{C}$ and ${\check{C}}^*$ associated with  $h$ are
\begin{align}
	\check{C}({h},t)
	&=
	\check{V}({h},t) - \sum_{j=1}^{d} \int_0^t {h}_j(u) \md \check{S}_j(u) + {A}^{\text{b}}(t), \\
	\check{C}^*(\check{h},t)
	&=
	\check{V}^*({h},t) - \sum_{j=1}^{d} \int_0^t {h}_j(u) \md \check{S}_j^*(u) + \check{A}^{\text{b},*}(t).
\end{align}
It follows from the results reviewed in Section~\ref{sec:rm} that the risk-minimizing investment strategy for the insurance payment process $A^{\text{b}}$ in the after-tax and after-expense market with numeraire ${\check{S}}_0$ can be expressed in terms of the Galtchouk-Kunita-Watanabe decomposition 
\begin{align}
 	\mathcal{V}^{\check{A}^{\text{b},*}}(t)
 	&=
 	\E^Q\!\left[ \left. \check{A}^{\text{b},*}(T) \, \right| \mathcal{F}(t)\right] \nonumber \\
 	&=
 	\mathcal{V}^{\check{A}^{\text{b},*}}(0)
 	+
 	\sum_{j=1}^d \int_0^t \check{h}^{\check{A}^{\text{b},*}}_j(u) \md \check{S}^*_j(u) + \check{L}^{\check{A}^{\text{b},*}}(t), \label{eq:GKW_Scheck}
\end{align}
where the zero-mean martingale $\check{L}^{\check{A}^{\text{b},*}}$ is orthogonal to the discounted after-tax and after-expense price processes $({\check{S}}^*_1,\ldots,{\check{S}}^*_d)$. 

With this notation in place, we can write up the following result, which is an immediate consequence of the results reviewed in Section~\ref{sec:rm}.

\begin{proposition}\label{prop:after-tax}
The unique risk-minimizing investment strategy ${\check{h}}^*$ in the after-tax and after-expense market is given by
\begin{align}\label{h_check_j}
	\check{h}^*_j(t) = \check{h}^{\check{A}^{\text{\em b},*}}_j(t)
\end{align}
for $j=1,\ldots, d$ and 
\begin{align}\label{h_check_0}
	\check{h}^*_0(t)  = \mathcal{V}^{\check{A}^{\text{\em b},*}}(t) - \check{A}^{\text{\em b},*}(t) - \sum_{j=1}^d \check{h}^*_j(t) \check{S}^*_j(t).
\end{align}
The associated value process is
\begin{align}\label{V_check_value}
	\check{V}(\check{h}^*,t)
	=
	\E^Q\!\left[ \left. \int_t^T e^{-\int_t^s \left(\left(1-\gamma(u)\right)r(u) - \delta(u)\right) \md u} \md A^{\text{\em b}}(s) \, \right| \mathcal{F}(t)\right]\!.
\end{align}
\end{proposition}
From~\eqref{V_check_value} we see that the current value of the investment strategy can be obtained as the conditional expected value of future payments, discounted with an artificial after-tax and after-expense short rate $((1-\gamma)r-\delta)$ rather than the original short rate $r$.

Recall that our interest lies in the before-tax and before-expense market rather than the artificial after-tax and after-expense market. We shall therefore now restate the risk-minimizing investment strategy in terms of quantities pertaining to the before-tax and before-expense market.

An alternative Galtchouk-Kunita-Watanabe decomposition of $\mathcal{V}^{\check{A}^{\text{b},*}}$ is obtained by ta-\linebreak king the discounted before-tax and before-expense price processes $({S}^*_1,\ldots,{S}^*_d)$ as integrators, which yields
\begin{align}\label{eq:gkw_altern}
 	\mathcal{V}^{\check{A}^{\text{b},*}}(t)
 	&=
 	\mathcal{V}^{\check{A}^{\text{b},*}}(0)
 	+
 	\sum_{j=1}^d \int_0^t {h}^{\check{A}^{\text{b},*}}_j(u) \md S^*_j(u) + {L}^{\check{A}^{\text{b},*}}(t), 
\end{align}
where the zero-mean martingale ${L}^{\check{A}^{*,\text{b}}}$ is orthogonal to the discounted before-tax and before-expense price processes $({S}^*_1,\ldots,{S}^*_d)$. It moreover follows from~\eqref{eq:dS*check} that
\begin{align}\label{eq:dS*check2}
		\md {S}^*_j(t)
		=
		\frac{{S}_j(t-)}{\check{S}_j(t-)} e^{-\int_0^t \left(\gamma(u)r(u) + \delta(u)\right) \md u} \frac{1}{1-\gamma(t-)}\! \md \check{S}^*_j(t).
\end{align} 
Because ${L}^{\check{A}^{*,\text{b}}}$ is orthogonal to $({S}^*_1,\ldots,{S}^*_d)$, we thus find that it is also orthorgonal to $(\check{S}^*_1,\ldots,\check{S}^*_d)$. By~\eqref{eq:dS*check2}, we also find that
\begin{align}
 	\mathcal{V}^{\check{A}^{\text{b},*}}(t)
 	=& \,
 	\mathcal{V}^{\check{A}^{\text{b},*}}(0)
 	+
 	{L}^{\check{A}^{\text{b},*}}(t) \nonumber \\
 	+& \,
 	\sum_{j=1}^d \int_0^t {h}^{\check{A}^{\text{b},*}}_j(u)
 	\frac{S_j(u-)}{\check{S}_j(u-)} e^{-\int_0^u \left(\gamma(\tau)r(\tau) + \delta(\tau)\right) \md \tau} \frac{1}{1-\gamma(u-)} \md \check{S}^*_j(u). \label{eq:gkw_altern_check}
\end{align}
Hence~\eqref{eq:gkw_altern_check} is another Galtchouk-Kunita-Watanabe decomposition of $\mathcal{V}^{\check{A}^{\text{b},*}}$ w.r.t.\ \linebreak  $(\check{S}^*_1,\ldots,\check{S}^*_d)$.  Uniqueness of the decomposition then implies $\check{L}^{\check{A}^{*,\text{b}}}={L}^{\check{A}^{*,\text{b}}}$ and
\begin{align}\label{eq:gkw_relation}
\frac{\check{S}_j(t-)}{S_j(t-)}
\check{h}^{\check{A}^{\text{b},*}}_j(t)
=
\frac{1}{1-\gamma(t-)} e^{-\int_0^t \left(\gamma(u)r(u) + \delta(u)\right) \md u} 
{h}^{\check{A}^{\text{b},*}}_j(t)
\end{align}
for $j=1,\ldots,d$.

While the risk-minimizing investment strategy of Proposition~\ref{prop:after-tax} pertains to the after-tax and after-expense market, it can be restated in terms of the before-tax and before-expense assets. Assume for a moment that the price processes are continuous, thus $S_j(t-)=S_j(t)$. An investment at time $t$ of $\check{h}^*_j(t)$ in the after-tax and after-expense asset $j$ corresponds to an investment of
\begin{align}
	h^*_j(t) = \frac{\check{S}_j(t)}{S_j(t)}  \check{h}^*_j(t)
\end{align}
in the before-tax and before-expense asset $j$. It follows from~\eqref{eq:gkw_relation} and~\eqref{h_check_j} that
\begin{align}
	h^*_j(t)
	=
	\frac{1}{1-\gamma(t-)} e^{-\int_0^t \left(\gamma(u)r(u) + \delta(u)\right) \md u}
	{h}^{\check{A}^{\text{b},*}}_j(t), \label{eq:relation1}
\end{align}
for $j=1,\ldots,d$, and from \eqref{h_check_0} and \eqref{eq:S_0_check} that
\begin{align}
	h^*_0(t)
	&=
	e^{-\int_0^t \left(\gamma(u)r(u) + \delta(u)\right) \md u}
	\left(
	\mathcal{V}^{\check{A}^{\text{b},*}}(t) - \check{A}^{\text{b},*}(t)
	\right)
	 - \frac{\check{S}_0(t)}{S_0(t)} \sum_{j=1}^d \check{h}^*_j(t) \check{S}^*_j(t) \nonumber \\
	&=
	 e^{-\int_0^t \left(\gamma(u)r(u) + \delta(u)\right) \md u}
	\left(
	\mathcal{V}^{\check{A}^{\text{b},*}}(t) - \check{A}^{\text{b},*}(t)
	\right)
	 - \sum_{j=1}^d h^*_j(t) S^*_j(t). \label{eq:relation2}
\end{align}
Even if the price processes are not continuous, we can consider the investment strategy ${h}^*$ given by
\begin{align}\label{eq:h_star_j}
	h^*_j(t) = \frac{\check{S}_j(t-)}{S_j(t-)}  \check{h}^*_j(t)
\end{align}
for $j=1,\ldots,d$ and
\begin{align}\label{eq:h_star_0}
	h^*_0(t) = \frac{\check{S}_0(t)}{S_0(t)}  \check{h}^*_0(t) + \left(S_0(t)\right)^{-1} \sum_{j=1}^d \check{h}^*_j(t) \check{S}_j(t)
	\left(
	1 - \frac{S_j(t)}{S_j(t-)} \frac{\check{S}_j(t-)}{\check{S}_j(t)}
	\right)\!,
\end{align}
where the investment in the before-tax and before-expense savings account exactly has been determined such that the total value remains unchanged, i.e.\ such that
\begin{align}
V(h^*,t)
=
\check{V}(\check{h}^*,t).
\end{align}
In particular, ${h}^*$ is $0$-admissible (pertaining to the before-tax and before-expense market). Furthermore, straightforward calculations show that ${h}^*$ also in the discontinuous case satisfies \eqref{eq:relation1} and \eqref{eq:relation2}. In the following subsection, we arrive at the same investment strategy by explicitly constructing the payment processes associated with taxes and expenses and applying the new concept of tax- and expense-modified risk-minimization rather than classic risk-minimization. 

\subsection{Tax- and expense-modified risk-minimization} \label{subsec:modified_riskmin}

In this subsection, we consider an alternative approach and construct explicitly payment processes related to taxes and expenses. Let $\gamma$ and $\delta$ be the tax and expense rates introduced in Subsection~\ref{subsec:after_tax}. Since the taxes and expenses lead to payments that depend on the investment strategy and the investment returns, we introduce two additional payment processes ${A}^{\text{tax}}({h})$ and  ${A}^{\text{e}}({h})$ for taxes and expenses, respectively, defined by ${A}^{\text{tax}}({h},0)=0$, ${A}^{\text{e}}({h},0)=0$, and
\begin{align}
	\md A^{\text{tax}}(h,t) &= \gamma(t-)  \sum_{j=0}^d h_j(t) \md S_j(t), \label{eq:dAtaxht}\\
	\md A^{\text{e}}(h,t) &= \delta(t) V(h,t) \md t. \label{eq:dAexpht}
\end{align}
We note that the taxes are symmetric in the sense that positive investment returns lead to a tax payment, whereas negative investment returns lead to a tax income (negative payment). We can interpret the taxes and expenses as negative dividends, by introducing dividend processes $D_j$ given by $D_j(0)=0$ and
\begin{align*}
\md D_j(t) = -\gamma(t-) \md S_j(t) - \delta(t) S_j(t) \md t
\end{align*}
for $j=0,\ldots,d$, when
\begin{align*}
\md A^{\text{tax}}(h,t)  + \md A^{\text{e}}(h,t) = - \sum_{j=0}^d h_j(t) \md D_j(t).
\end{align*}
The traded assets with price processes $(S_0,S_1,\ldots,S_d)$ are then to be seen as trading ex dividend (or before taxes and expenses). In comparison, the artificial market with price processes $(\check{S}_0,\check{S}_1,\ldots,\check{S}_d)$ of Subsection~\ref{subsec:after_tax} was given the interpretation of being after taxes and expenses. As mentioned in Section~\ref{sec:intro}, risk-minimization in the presence of dividends appears to be a rather unexplored area of research; for a general introduction to dividends in continuous time we refer to \cite{bjork2009} Chapter 16.

We are interested in the problem of determining risk-minimizing investment strategies for the combined payments consisting of the three payment processes ${A}^{\text{b}}$, ${A}^{\text{tax}}({h})$, and ${A}^{\text{e}}({h})$. Thus, we define the undiscounted cost process in the presence of taxes and expenses as the cost process of the total payments, which depend on the choice of our object of interest, namely the investment strategy $h$. In this sense, we are facing a fixed-point problem.
\begin{definition}The undiscounted cost process ${C}$ in the presence of taxes and expenses associated with an investment strategy ${h}$ and an insurance payment process ${A}^{\text{\em b}}$ is defined by
\begin{align}\label{eq:Cwithtax}
	C(h,t) = V(h,t) - \sum_{j=0}^d \int_0^t h_j(u) \md S_j(u) + A^{\text{\em b}}(t) + A^{\text{\em tax}}(h,t) + A^{\text{\em e}}(h,t),
\end{align}
where  ${A}^{\text{\em tax}}({h})$ and ${A}^{\text{\em e}}({h})$ are defined by~(\ref{eq:dAtaxht}) and~(\ref{eq:dAexpht}), respectively.
\end{definition}
We see that $C(h,t)$ comprises the accumulated costs during $[0,t]$ including the payments ${A}^{\text{b}}(t)$, ${A}^{\text{tax}}({h},t)$, and ${A}^{\text{e}}({h},t)$. Therefore, the value process at time $t$, i.e.\ $V(h,t)$, should, in a similar fashion to previously, be interpreted as the value of the portfolio $h$ held at time $t$ \textit{after} all payments during $[0,t]$, \textit{including taxes and expenses}.

The discounted cost process $C^*$ is defined from the undiscounted cost process via~\eqref{eq:dC*htdef}, i.e.\
\begin{align}\label{eq:defC*}
	C^*(h,t) = C(h,0) +  \int_0^t S_0^{-1}(u) \md C(h,u).
\end{align}

We now introduce the following definitions of tax- and expense-modified value, cost, and risk processes, respectively.

\begin{definition}\label{def:modifiedValueAndCost}The tax- and expense-modified value and cost processes are defined via 
\begin{align}
	\tilde{V}(h,t) &= e^{\int_0^t \left(\gamma(u) r(u) + \delta(u)\right) \md u} {V}^*(h,t), \label{eq:Vtildeht}\\
	\tilde{C}(h,t) &= C^*(h,0) + \int_0^t  e^{\int_0^u  \left(\gamma(\tau) r(\tau) + \delta(\tau)\right) \md \tau} \md C^*(h,u), \label{eq:Ctildeht}
\end{align}
where ${C}^*$ is defined by~\eqref{eq:defC*}, and where ${C}$ is defined by~\eqref{eq:Cwithtax}. A strategy is said to be risk-minimizing for ${A}^{\text{\em b}}$ in the presence of taxes and expenses if it is $0$-admissible and minimizes for all $t\in[0,T]$ the tax- and expense-modified risk process ${\tilde{R}}$ defined by 
\begin{align}\label{eq:Rtildeht}
	\tilde{R}(h,t) = \E^Q\!\left[\left.  \left(\tilde{C}(h,T)- \tilde{C}(h,t) \right)^2 \right| \mathcal{F}(t) \right]\!.
\end{align}
\end{definition}
Note that the tax- and expense-modified quantities correspond to the usual discounted quantities but using as numeraire the after-tax and after-expense savings account rather than the before-tax and before-expense savings account.

Using methods similar to the ones applied in~\cite{buchardtmoeller2018}, we observe that
\begin{align*}
	e^{\int_0^t  \left(\gamma(u) r(u) + \delta(u)\right) \md u}&  \md C^*(h,t) \\  =& \md \!\left(  e^{\int_0^t  \left(\gamma(u) r(u) + \delta(u)\right) \md u} V^*(h,t)\right)   + 
	e^{\int_0^t (\gamma(u) r(u) + \delta(u)) \md u} \md A^{\text{b},*}(t) \\
	 -& \sum_{j=1}^d h_j(t) \left(1-\gamma(t-)\right) e^{\int_0^t  \left(\gamma(u) r(u) + \delta(u)\right) \md u} \md S^*_j(t).   
\end{align*}
It follows from the definition given by~\eqref{eq:Ctildeht} and the above calculation that 
\begin{align}
	\md \tilde{C}(h,t) &= \md \tilde{V}(h,t) - \sum_{j=1}^d h_j(t) \left(1-\gamma(t-)\right) e^{\int_0^t  \left(\gamma(u) r(u) + \delta(u)\right) \md u} \md S^*_j(t)  \nonumber \\
	 & + e^{\int_0^t  \left(\gamma(u) r(u) + \delta(u)\right) \md u} \md A^{\text{b},*}(t). \label{eq:dCtildeht}
\end{align}
Thus, the dynamics of the tax- and expense-modified cost process ${\tilde{C}}$ has a structure which is similar to the dynamics for the discounted cost process in the traditional setting, compare with~\eqref{eq:dC*ht}. This observation enables us to use similar techniques as in the classic setting without taxes and expenses for determining risk-minimizing investment strategies. To do this, we first define a tax- and expense-modified version ${\tilde{A}}^{\text{b}}$ of the discounted insurance payment process ${A}^{\text{b},*}$ via 
\begin{align*}
	\tilde{A}^{\text{b}}(t) &= A^{\text{b},*}(0) + \int_0^t  e^{\int_0^s  \left(\gamma(u) r(u) + \delta(u)\right) \md u} \md A^{\text{b},*}(s).
\end{align*}
Note that the tax- and expense-modified insurance payment process corresponds to the usual discounted insurance payment process but with the after-tax and after-expense savings account rather than the before-tax and before-expense savings account as numeraire, i.e.
\begin{align}\label{eq:tilde_check_ident}
\tilde{A}^{\text{b}}
=
\check{A}^{\text{b},*},
\end{align}
where ${\check{A}}^{\text{b},*}$ is given by~\eqref{eq:dAcheck*}. The new notation is solely to stress the change in interpretation compared to the Subsection~\ref{subsec:after_tax}.

We again define the intrinsic value process associated with $\tilde{A}^{\text{b}}$ as the $Q$-martingale 
\begin{align}\label{eq:VAtildeT}
	\mathcal{V}^{\tilde{A}^{\text{b}}}(t) 
	&= \E^Q\!\left[ \left. \tilde{A}^{\text{b}}(T) \right| \mathcal{F}(t)\right]  \\
	&= \tilde{A}^{\text{b}}(t) + 
	\E^Q\!\left[ \left. \int_t^T e^{-\int_0^s r(u) \md u} e^{\int_0^s  \left(\gamma(u) r(u) + \delta(u)\right) \md u} \md A^{\text{b}}(s) \, \right| \mathcal{F}(t)\right]\!, \label{eq:VAtildeT2}
\end{align}
and (re)write its Galtchouk-Kunita-Watanabe decomposition as
 \begin{align}\label{eq:VAtildeTdec}
	\mathcal{V}^{\tilde{A}^{\text{b}}}(t) = \mathcal{V}^{\tilde{A}^{\text{b}}}(0) + \sum_{j=1}^d \int_0^t h^{\tilde{A}^{\text{b}}}_j(u) \md S^*_j(u) + L^{\tilde{A}^{\text{b}}}(t).
\end{align}
Due to uniqueness of the decomposition, we have $\mathcal{V}^{\tilde{A}^{\text{b}}}=\mathcal{V}^{\check{A}^{\text{b},*}}$, $L^{\tilde{A}^{\text{b}}}=L^{\check{A}^{\text{b},*}}$, and $h^{\tilde{A}^{\text{b}}}_j = h^{\check{A}^{\text{b},*}}_j$ for $j=1,\ldots,d$, confer also with~\eqref{eq:tilde_check_ident} and~\eqref{eq:gkw_altern}.

The following theorem contains the main result of the paper.

\begin{theorem}\label{te:rmtax}There exists a unique risk-minimizing investment strategy ${\tilde{h}}$ for ${A}^{\textit{\em b}}$ in the presence of taxes and expenses given by
\begin{align}\label{eq:htilde_j_opt}
	\tilde{h}_j(t) = \frac{1}{1-\gamma(t-)}e^{-\int_0^t \left(\gamma(u) r(u) + \delta(u)\right) \md u} h^{\tilde{A}^{\text{\em b}}}_j(t),
\end{align}
for $j=1,\ldots, d$ and 
\begin{align}\label{eq:htilde_0_opt}
	\tilde{h}_0(t)  =  e^{-\int_0^t \left(\gamma(u) r(u) + \delta(u)\right) \md u}\left(\mathcal{V}^{\tilde{A}^{\text{\em b}}}(t) - \tilde{A}^{\text{\em b}}(t)\right)  - \sum_{j=1}^d \tilde{h}_j(t) S^*_j(t).
\end{align}
The associated risk process is given by
\begin{align}\label{eq:Rtildeht_optimal}
	\tilde{R}(\tilde{h},t) = \E^Q\!\left[ \left. {\left( L^{\tilde{A}^{\text{\em b}}}\!(T)-L^{\tilde{A}^{\text{\em b}}}\!(t)  \right)}^2 \, \right| \mathcal{F}(t)\right]\!,
\end{align}
and the associated value process is
\begin{align}\label{V_value}
	V(\tilde{h},t)
	=
	\E^Q\!\left[ \left. \int_t^T e^{-\int_t^s \left(\left(1-\gamma(u)\right)r(u) - \delta(u)\right) \md u} \md A^{\text{\em b}}(s) \, \right| \mathcal{F}(t)\right]\!.
\end{align}
\end{theorem} 
The proof of the result is presented below. First, we give an interpretation of the result and relate it to the risk-minimizing investment strategy in the after-tax and after-expense market, confer with Proposition~\ref{prop:after-tax} and the discussion thereafter.

By comparing~(\ref{eq:VAtildeTdec}) and~(\ref{eq:VA*Tdec}), we see that the quantity $h^{{\tilde{A}}^{\text{\em b}}}_j(t)$ can be interpreted as the number of assets $j$ at time $t$ in a risk-minimizing investment strategy for the modified payment process $\tilde{{A}}^{\text{b}}$ in the classic setting without taxes and expenses, see Section~\ref{sec:rm}. The solution of Theorem~\ref{te:rmtax} is a modification of this strategy which adjusts for the taxes and expenses via the factors $\left(1-\gamma(t-)\right)^{-1}$ and $e^{-\int_0^t \left(\gamma(u) r(u) + \delta(u)\right) \md u}$.

From~\eqref{V_value} we see that the current value of the investment strategy can be obtained as the conditional expected value of future payments given the information currently available but discounted with a modified short rate $((1-\gamma)r-\delta)$ rather than the original short rate $r$. 

The modified discount factor corresponds to using the after-tax and after-expense savings account rather than the before-tax and before-expense savings account as the reference for the time-value of money. Moreover, the value agrees with the value obtained in Subsection~\ref{subsec:after_tax} for the after-tax and after-expense market, confer with~\eqref{V_check_value}, i.e.\
\begin{align*}
	V(\tilde{h},t)
	=
	\check{V}(\check{h}^*,t)
\end{align*}
with ${\check{h}}^*$ given by \eqref{h_check_j} and~\eqref{h_check_0}. Actually, we see that ${\tilde{h}} = {h^*}$ with ${h^*}$ given by~\eqref{eq:h_star_j} and~\eqref{eq:h_star_0}. This proves the relation between after-tax and after-expense risk-minimiza-\ tion and tax- and expense-modified risk-minimization alluded to at the end of Subsection~\ref{subsec:after_tax}. In other words, strategies resulting from the two conceptually different approaches are mathematically equivalent in the sense that the sums invested in each asset are equal.

The proof of Theorem~\ref{te:rmtax} can in principle be based on Proposition~\ref{prop:after-tax}. Straightforward calculations show that the modified and discounted cost processes $\tilde{C}$ and $\check{C}^*$, respectively, are directly related by adjusting the investment strategies according to the mappings given by  \eqref{eq:h_star_j} and \eqref{eq:h_star_0}. Thus by essentially combining \eqref{eq:gkw_relation} and \eqref{eq:gkw_altern} with the identity \eqref{eq:tilde_check_ident}, the proof can be established. To better reveal what is going on behind the scenes, we provide a direct proof.
\begin{proof}[Proof of Theorem~\ref{te:rmtax}]
The result is proven by determining the quantity~\eqref{eq:Rtildeht} and minimizing it. Since by definition ${A}^{\text{tax}}(h,0)=0$ and ${A}^{\text{e}}(h,0)=0$, we see that $\tilde{C}({h},0) = \tilde{V}({h},0) + \tilde{A}^{\text{b}}(0)$. It now follows from~\eqref{eq:dCtildeht} and the definition of the modified cost process that
\begin{align}
	\tilde{C}({h},t) 
	&= \tilde{V}({h},t) + \tilde{A}^{\text{b}}(t) \nonumber \\
	&- \sum_{j=1}^d \int_0^t \left(1-\gamma(u-)\right) {h}_j(u) e^{\int_0^u \left(\gamma(\tau)r(\tau) + \delta(\tau)\right) \md \tau} \md S^*_j(u). \label{eq:tildeC_defa}
\end{align}
For $0$-admissible strategies ${\tilde{h}}$ we have that $V(\tilde{h},T)=0$ and hence also $\tilde{V}(\tilde{h},T)=0$.  
Moreover, (\ref{eq:VAtildeTdec}) implies that 
\begin{align}\label{eq:VAtildeT3}
	\tilde{A}^{\text{b}}(T) = \mathcal{V}^{\tilde{A}^{\text{b}}}(T) = \mathcal{V}^{\tilde{A}^{\text{b}}}(0) + \sum_{j=1}^d \int_0^T h^{\tilde{A}^{\text{b}}}_j(u) \md S^*_j(u) + L^{\tilde{A}^{\text{b}}}(T). 
\end{align}	
Now insert~\eqref{eq:VAtildeT3} into~\eqref{eq:tildeC_defa} with $t=T$ and use $\tilde{V}(\tilde{h},T)=0$ to see that
\begin{align*}
	\tilde{C}(\tilde{h},T) \nonumber
	&= \mathcal{V}^{\tilde{A}^{\text{b}}}(0) + L^{\tilde{A}^{\text{b}}}(T) \nonumber \\
	&+\sum_{j=1}^d \int_0^T \left( h^{\tilde{A}^{\text{b}}}_j(u) -  \left(1-\gamma(u-)\right) \tilde{h}_j(u) e^{\int_0^u \left(\gamma(\tau)r(\tau) + \delta(\tau)\right) \md \tau}\right) \! \md S^*_j(u) \nonumber \\
	&= \mathcal{V}^{\tilde{A}^{\text{b}}}(t) +  \left( L^{\tilde{A}^{\text{b}}}(T) -  L^{\tilde{A}^{\text{b}}}(t) \right) \nonumber \\ 
	&- \sum_{j=1}^d \int_0^t  \left(1-\gamma(u-)\right) \tilde{h}_j(u) e^{\int_0^u \left(\gamma(\tau)r(\tau) + \delta(\tau)\right) \md \tau} \md S^*_j(u)  \nonumber \\
	&+ \sum_{j=1}^d \int_t^T \left( h^{\tilde{A}^{\text{b}}}_j(u) -  \left(1-\gamma(u-)\right) \tilde{h}_j(u) e^{\int_0^u \left(\gamma(\tau)r(\tau) + \delta(\tau)\right) \md \tau}\right) \! \md S^*_j(u),
\end{align*}
where the last equality follows from~\eqref{eq:VAtildeTdec}. By combining the expressions for $\tilde{C}(\tilde{h},T)$ and $\tilde{C}(\tilde{h},t)$ and by using the orthogonality of the $Q$-martingales $S^*_j$ and $L^{\tilde{A}^{\text{b}}}$, we find that
\begin{align*}
	\tilde{R}(\tilde{h},t) &= R_1(\tilde{h},t)  + \E^Q\!\left[\left.\left( L^{\tilde{A}^{\text{b}}}(T) -  L^{\tilde{A}^{\text{b}}}(t) \right)^2 \,  \right| \mathcal{F}(t)\right]
	+ R_2(\tilde{h},t)
\end{align*}
with $R_1$ and $R_2$ given by
\begin{align*}
&R_1(\tilde{h},t)
=
 \left( \mathcal{V}^{\tilde{A}^{\text{b}}}(t) - \tilde{V}(\tilde{h},t) - \tilde{A}^{\text{b}}(t)  \right)^2 \\
&R_2(\tilde{h},t) \\
&=
\E^Q\!\left[\left.
\left(\sum_{j=1}^d \int_t^T \left(
 h^{\tilde{A}^{\text{b}}}_j(u) -  \left(1-\gamma(u-)\right) \tilde{h}_j(u) e^{\int_0^u \left(\gamma(\tau)r(\tau) + \delta(\tau)\right) \md \tau}
\right)\! \md S^*_j(u)    \right)^{\!\!2}  \,  \right| \mathcal{F}(t)\right]\!.
\end{align*}
The terms $R_1$ and $R_2$ can now be eliminated as follows. First, the term $R_2$ is eliminated by e.g.\ choosing $(\tilde{h}_1,\ldots,\tilde{h}_d)$ according to~\eqref{eq:htilde_j_opt}. By examining $R_1$, we realize that this term next can be eliminated if and only if
\begin{align*}
	e^{\int_0^t \left(\gamma(\tau)r(\tau) + \delta(\tau)\right) \md \tau} V^*(\tilde{h},t) = \tilde{V}(\tilde{h},t) = \mathcal{V}^{\tilde{A}^{\text{b}}}(t) -  \tilde{A}^{\text{b}}(t), 	
\end{align*}
i.e.\ if and only if
\begin{align*}
	V^*(\tilde{h},t) = e^{-\int_0^t \left(\gamma(\tau)r(\tau) + \delta(\tau)\right) \md \tau}
	\left( \mathcal{V}^{\tilde{A}^{\text{b}}}(t) - \tilde{A}^{\text{b}}(t)\right)\!,
\end{align*}
which is then obtained by choosing $\tilde{h}_0$ uniquely according to~\eqref{eq:htilde_0_opt}. This shows that $\tilde{h}$ given by~\eqref{eq:htilde_j_opt} and~\eqref{eq:htilde_0_opt} is a risk-minimizing investment strategy for ${A}^{\textit{\em b}}$ in the presence of taxes and expenses, and, further, establishes~\eqref{eq:Rtildeht_optimal} and~\eqref{V_value}.

To prove uniqueness, it remains to be shown that each $\tilde{h}_j$ for $j=1,\ldots,d$ is uniquely determined. First note that by Itô isometry and~\eqref{eq:sigmaB},
\begin{align*}
R_2(\tilde{h},0)
&=
\sum_{i=1}^d \sum_{j=1}^d  
\E^Q\!\left[  \int_0^T \alpha_i(t)\alpha_j(t) \md \langle S_i^*,S_j^*\rangle(t)
\right]\\
&=
\E^Q\!\left[\int_0^T \alpha^{\text{tr}}(t) \sigma_S(t) \alpha(t) \md B(t)
\right]\!,
\end{align*}
where $\alpha=(\alpha_1,\ldots,\alpha_d)$ is given by
\begin{align*}
\alpha_j(t) = h^{\tilde{A}^{\text{b}}}_j(t) -  \left(1-\gamma(t-)\right) \tilde{h}_j(t) e^{\int_0^t \left(\gamma(u)r(u) + \delta(u)\right) \md u}
\end{align*}
for $j=1,\ldots,d$. Recall that each $\sigma_S(t)$ is assumed positive definite and that $B$ is null at $0$ and strictly increasing. It follows that $R_2(\tilde{h},0)=0$ if and only if each $\alpha_j$ is zero, i.e.\ if and only if $(\tilde{h}_1,\ldots,\tilde{h}_d)$ is chosen according to~\eqref{eq:htilde_j_opt}. This proves uniqueness of the risk-minimizing investment strategy.
\end{proof}

\subsection{Two-step risk-minimization} \label{subsec:two-step}

In this subsection, we study two-step risk-minimization in the following sense. Assume that the original investor, say an insurer, adopts the risk-minimizing investment strategy ${\tilde{h}}$ in the presence of taxes and expenses given by Theorem~\ref{te:rmtax}. Then it faces the tax specific payments ${A}^{\text{tax}}({\tilde{h}})$ and expense payments ${A}^{\text{e}}({\tilde{h}})$ in addition to the original insurance payments ${A^{\!\text{b}}}$. We consider the scenario where a systemic investor, e.g.\ a re-insurer, assumes the payments and faces the problem of classic risk-minimization (in the absence of taxes and expenses) within the setup of Section~\ref{sec:rm}. The relevant classic insurance payment process ${A}$ is thus given by
\begin{align}\label{eq:payments_twostep}
	A(t)
	=
	A^{\text{b}}(t)
	+
	A^{\text{tax}}(\tilde{h},t)
	+
	A^{\text{e}}(\tilde{h},t),
\end{align}
where ${\tilde{h}}$ is determined by Theorem~\ref{te:rmtax} and thus fixed. To determine the classic risk-minimizing investment strategy for ${A}$, we now study the intrinsic value process $\mathcal{V}^{A^*}$ associated with $A^*$. The following lemma is key.

\begin{lemma}\label{lemma:intrinsic_twostep}
With ${A}$ given by~\eqref{eq:payments_twostep}, the corresponding intrinsic value process $\mathcal{V}^{A^*}$ has the following Galtchouk-Kunita-Watanabe decomposition:
\begin{align*}
\mathcal{V}^{A^*}(t)
	&=
	\mathcal{V}^{A^*}(0)
	+
	\sum_{j=1}^d \int_0^t \tilde{h}_j(u) \md S^*_j(u) + \int_0^t  e^{-\int_0^u \left(\gamma(\tau)r(\tau) + \delta(\tau)\right) \md \tau} \md L^{\tilde{A}^{\text{\em b}}}(u),
\end{align*}
with ${\tilde{h}}$ the risk-minimizing investment strategy of ${{A}^{\text{\em b}}}$ in the presence of taxes and expenses given by Theorem~\ref{te:rmtax} and with $L^{\tilde{A}^{\text{\em b}}}$ as in the Galtchouk-Kunita-Watanabe decomposition of $\mathcal{V}^{\tilde{A}^{\!\text{\em b}}}$, confer with \eqref{eq:VAtildeTdec}. Furthermore,
\begin{align*}
\mathcal{V}^{A^*}(t)
	&=
	A^*(t) + e^{-\int_0^t \left(\gamma(u)r(u) + \delta(u)\right) \md u}
	\left( \mathcal{V}^{\tilde{A}^{\text{\em b}}}(t) - \tilde{A}^{\text{\em b}}(t)\right)\!.
\end{align*}
\end{lemma}

\begin{proof}
By straightforward calculations we obtain the following expression:
\begin{align*}
	A^*(t)
	&=
	\int_0^t e^{-\int_0^u r(\tau) \md \tau} \md A(u) \\
	&=
	A^{\text{b},*}(t)
	+
	\sum_{j=1}^d
	\int_0^t \gamma(u-) \tilde{h}_j(u) \md S^*_j(u) \\
	&\hspace{16.6mm}+
	\int_0^t \left(\gamma(u)r(u) + \delta(u)\right) e^{-\int_0^u r(\tau) \md \tau} V(\tilde{h},u) \md u.
\end{align*}
Because the discounted price processes are $Q$-martingales, it follows that
\begin{align*}
	\mathcal{V}^{A^*}(t)
	&=
	\E^Q\!\left[ \left. A^*(T) \, \right| \mathcal{F}(t)\right] \\
	&=
	\mathcal{V}^{A^{\text{b},*}}(t)
	+
	A^{\text{tax},*}(\tilde{h},t)
	+
	A^{\text{e},*}(\tilde{h},t) \\
	&+
	\E^Q\!\left[ \left. \int_t^T \left(\gamma(u)r(u) + \delta(u)\right) e^{-\int_0^u r(\tau) \md \tau} V(\tilde{h},u) \md u\, \right| \mathcal{F}(t)\right].
\end{align*}
From \eqref{V_value},
\begin{align*}
	V(\tilde{h},u)
	=
	\E^Q\!\left[ \left. \int_u^T e^{-\int_u^s \left(\left(1-\gamma(\tau)\right)r(\tau) - \delta(\tau)\right) \md \tau} \md A^{\text{b}}(s) \, \right| \mathcal{F}(u)\right]\!.
\end{align*}
By the law of iterated expectations and by interchanging the order of integration, simple manipulations yield
\begin{align*}
	\hspace{10mm}&\hspace{-10mm}\E^Q\!\left[ \left. \int_t^T \left(\gamma(u)r(u) + \delta(u)\right) e^{-\int_0^u r(\tau) \md \tau} V(\tilde{h},u) \md u\, \right| \mathcal{F}(t)\right] \\
	&=
	\E^Q\!\left[ \left. \int_t^T e^{-\int_0^s r(\tau) \md\tau} 
	\int_t^s 
	\left(\gamma(u)r(u) + \delta(u)\right) e^{\int_u^s  \left(\gamma(\tau)r(\tau) + \delta(\tau)\right) \md \tau} \md u
	\md A^{\text{b}}(s)
	\, \right| \mathcal{F}(t)\right] \\
	&=
	\E^Q\!\left[ \left. \int_t^T e^{-\int_0^s r(\tau) \md\tau} 
	\left(e^{\int_t^s  \left(\gamma(\tau)r(\tau) + \delta(\tau)\right) \md \tau} - 1\right) \!\md A^{\text{b}}(s)
	\, \right| \mathcal{F}(t)\right] \\
	&=
	e^{-\int_0^t  \left(\gamma(\tau)r(\tau) + \delta(\tau)\right) \md \tau}\left(\mathcal{V}^{\tilde{A}^{\text{b}}}(t) - \tilde{A}^{\text{b}}(t)\right)
	-
	\left(\mathcal{V}^{{A}^{\text{b},*}}(t) - {A}^{\text{b},*}(t)\right),
\end {align*}
see also~\eqref{eq:VA*T} and~\eqref{eq:VAtildeT2}. Collecting everything, we obtain
\begin{align*}
	\mathcal{V}^{A^*}(t)
	&=
	\mathcal{V}^{A^{\text{b},*}}(t)
	+
	A^{\text{tax},*}(\tilde{h},t)
	+
	A^{\text{e},*}(\tilde{h},t) \\
	&+
	e^{-\int_0^t  \left(\gamma(\tau)r(\tau) + \delta(\tau)\right) \md \tau}\left(\mathcal{V}^{\tilde{A}^{\text{b}}}(t) - \tilde{A}^{\text{b}}(t)\right)
	-
	\left(\mathcal{V}^{{A}^{\text{b},*}}(t) - {A}^{\text{b},*}(t)\right) \\
	&=
	A^{\text{b},*}(t)
	+
	A^{\text{tax},*}(\tilde{h},t)
	+
	A^{\text{e},*}(\tilde{h},t)
	+
	e^{-\int_0^t  \left(\gamma(\tau)r(\tau) + \delta(\tau)\right) \md \tau}\left(\mathcal{V}^{\tilde{A}^{\text{b}}}(t) - \tilde{A}^{\text{b}}(t)\right) \\
	&=
	A^*(t)
	+
	e^{-\int_0^t  \left(\gamma(\tau)r(\tau) + \delta(\tau)\right) \md \tau}\left(\mathcal{V}^{\tilde{A}^{\text{b}}}(t) - \tilde{A}^{\text{b}}(t)\right)\!,
\end{align*}
as desired. Now using integration by parts and the definition of ${\tilde{A}^{\text{b}}}$, we find that
\begin{align*}
\md \mathcal{V}^{A^*}(t)
=
\sum_{j=1}^d \gamma(t-) \tilde{h}_j(t) \md S^*_j(t)
+
e^{-\int_0^t  \left(\gamma(\tau)r(\tau) + \delta(\tau)\right) \md \tau} \md \mathcal{V}^{\tilde{A}^{\text{b}}}(t).
\end{align*}
From the Galtchouk-Kunita-Watanabe decomposition of $\mathcal{V}^{\tilde{A}^{\text{b}}}$, confer with~\eqref{eq:VAtildeTdec}, and the identity~\eqref{eq:htilde_j_opt}, it follows that
\begin{align*}
\md \mathcal{V}^{A^*}(t)
&=
\sum_{j=1}^d \gamma(t-) \tilde{h}_j(t) \md S^*_j(t)
+
\sum_{j=1}^d (1-\gamma(t-)) \tilde{h}_j(t) \md S^*_j(t) \\
&+
e^{-\int_0^t  \left(\gamma(\tau)r(\tau) + \delta(\tau)\right) \md \tau} \md L^{\tilde{A}^{\text{b}}}(t) \\
&=
\sum_{j=1}^d \tilde{h}_j(t) \md S^*_j(t)
+
e^{-\int_0^t  \left(\gamma(\tau)r(\tau) + \delta(\tau)\right) \md \tau} \md L^{\tilde{A}^{\text{b}}}(t).
\end{align*}
Note that the final term is a zero-mean $Q$-martingale orthogonal to the discounted price processes, because this is the case for $L^{\tilde{A}^{\text{b}}}$. We conclude that $\mathcal{V}^{A^*}$ has Galtchouk-Kunita-Watanabe decomposition given by
\begin{align*}
	\mathcal{V}^{A^*}(t)
	&=
	\mathcal{V}^{A^*}(0)
	+
	\sum_{j=1}^d \int_0^t \tilde{h}_j(u) \md S^*_j(u) + \int_0^t  e^{-\int_0^u \left(\gamma(\tau)r(\tau) + \delta(\tau)\right) \md \tau} \md L^{\tilde{A}^{\text{\em b}}}(u),
\end{align*}
as desired.
\end{proof}

Combining Lemma~\ref{lemma:intrinsic_twostep} with the classic results reviewed in Section~\ref{sec:rm}, we obtain the main result of this subsection:

\begin{proposition}\label{prop:main_twostep}
With ${A}$ given by~\eqref{eq:payments_twostep} there exists a unique classic risk-minimizing investment strategy for ${A}$, and this investment strategy is identical to the unique risk-minimizing investment strategy for ${A}^{\!\textit{\em b}}$ in the presence of taxes and expenses.
\end{proposition}

\begin{proof}

It follows from the results reviewed in Section~\ref{sec:rm} that there exists a unique classic risk-minimizing investment strategy ${h}^*$ for ${A}$ given by $h_j^*= \tilde{h}_j$ for $j=1,\ldots,d$, due to Lemma~\ref{lemma:intrinsic_twostep}, and
\begin{align*}
h_0^*(t)
=
\mathcal{V}^{A^*}(t)
-
A^*(t)
-
\sum_{j=1}^d  \tilde{h}_j(t) S_j^*(t).
\end{align*}
To establish the theorem, it remains to be shown that $h_0^*(t) = \tilde{h}_0(t)$. From~\eqref{eq:htilde_0_opt},
\begin{align*}
	\tilde{h}_0(t)  =  e^{-\int_0^t \left(\gamma(u) r(u) + \delta(u)\right) \md u}\left(\mathcal{V}^{\tilde{A}^{\text{\em b}}}(t) - \tilde{A}^{\text{\em b}}(t)\right)  - \sum_{j=1}^d \tilde{h}_j(t) S^*_j(t),
\end{align*}
such that it suffices to show that
\begin{align*}
\mathcal{V}^{A^*}(t) - A^*(t)
=
 e^{-\int_0^t \left(\gamma(u) r(u) + \delta(u)\right) \md u}\left(\mathcal{V}^{\tilde{A}^{\text{\em b}}}(t) - \tilde{A}^{\text{\em b}}(t)\right).
\end{align*}
But this also immediately follows from Lemma~\ref{lemma:intrinsic_twostep} thus completing the proof.
\end{proof}

Proposition~\ref{prop:main_twostep} allows us to draw the following conclusion. If a systemic investor assumes all payments, including taxes and expenses, of an original investor adopting the risk-minimizing investment strategy in the presence of taxes and expenses, and faces the problem of classic risk-minimization (in the absence of taxes and expenses), then the optimal strategy coincides with the investment strategy adopted by the original investor; in particular, additional risk reduction is impossible, and in this specific sense, tax- and expense-modified risk-minimization is consistent with classic risk-minimization.

\section{Case study: classic multi-state life insurance payments}\label{sec:ex}

We extend the classic life insurance setting, see e.g.\ \cite{hoem69,norberg1991,christiansen2012}, by allowing for investments in a bond market following a Vasicek term structure model with a deterministic tax rate and expenses depending on the current state of the insurance contract(s). Our setup and results are similar to earlier examples from the literature on risk-minimization, see e.g.~\cite{thmoller01} Subsection 3.1, with the primary new contribution being the inclusion of taxes and expenses; in particular, we derive the risk-minimizing investment strategy in the presence of taxes and expenses using the tools developed in Section~\ref{sec:tmrm}. Throughout the exposition, we explain how to extend the results to general term structure models.

In Subsection~\ref{subsec:ex_setup}, we introduce models for the market, insurance payment process, and taxes and expenses. Next, the risk-minimizing investment strategy in the presence of taxes and expenses is derived in Subsection~\ref{subsec:ex_riskmin}. Finally, we discuss valuation and computability with a view towards actuarial practice in Subsection~\ref{subsec:ex_dis}.

\subsection{Setup}\label{subsec:ex_setup}

The financial market consists of two assets with price processes $(S_0,S_1)$ driven by a stochastic short rate process $r$ following a Vasicek model. In other words, $r$ is an Ornstein-Uhlenbeck process satisfying the stochastic differential equation
\begin{align*}
\md r(t)
=
\kappa\left(\theta-r(t)\right) \md t + \sigma \md W(t),
\end{align*}
where $\kappa$, $\theta$, and $\sigma$ are positive constants and $W$ is a standard Brownian motion under an equivalent martingale measure $Q$.

The development of an underlying life insurance contract (or multiple contracts) is described by the classic multi-state Markov model of~\cite{hoem69}. Let $Z$ be a Markovian jump process with values in a finite set $\mathcal{J}=\{0,1,\ldots,J\}$ describing the state of the contract(s). The initial state of the contract(s) is taken to be $0$ such that $Z(0)=0$.

A multivariate counting process $N=(N_{jk})_{j,k \in \mathcal{J},k \neq j}$ is associated with the jump process $Z$ by setting $N_{jk}(0)=0$ and
\begin{align*}
N_{jk}(t) = \#\!\left\{s \in (0,t] \, : Z(s-) = j, Z(s) = k \right\} 
\end{align*}
for $t \in (0,T]$. The quantities $N_{jk}(t)$ can be interpreted as the number of transitions from state $j$ to state $k$ of the contract(s) within the time interval $[0,t]$.

We assume that $Z$ and the financial market given by $W$ are independent under $Q$, and we take the filtration $\mathbb{F}$ to be the $Q$-augmentation of the natural filtration of $Z$ and $W$.

As before, $S_0$ is the savings account taking the form
\begin{align*}
S_0(t) =\exp\!\left(\int_0^t r(u) \md u\right)\!.
\end{align*}
In addition to the savings account, the market contains a zero coupon bond with expiry at time $T>0$. The price process is:
\begin{align*}
S_1(t)
=
\E^Q\!\left[\left. \frac{S_0(t)}{S_0(T)} \, \right| \mathcal{F}(t) \right]
=
\E^Q\!\left[\left.  e^{-\int_t^T r(s) \md s} \, \right| \mathcal{F}(t) \right]\!.
\end{align*}
We assume there exists continuous functions $[0,T] \ni t \mapsto \mu_{jk}(t)$, $j,k \in \mathcal{J}$, $k \neq j$, such that $Z$ has transition intensities $\mu$ completely characterizing the distribution of $Z$. It follows that the processes $M_{jk}$ given by
\begin{align*}
M_{jk}(t)
=
N_{jk}(t)
-
\int_0^t \1{Z(s-) = j} \mu_{jk}(s) \md s
\end{align*}
are orthogonal martingales. Furthermore, each $M_{jk}$ is also orthogonal to the discounted price process $S_1^*=S_1/S_0$.

We are interested in payment processes related to the development of the insurance contract(s). Specifically, the insurance payment process $A^{\text{b}}$ has dynamics
\begin{align}\label{eq:ex_payments}
\md A^{\text{b}}(t)
=
\sum_{j \in \mathcal{J}} \left( \1{Z(t) = j} b_j(t) \md t + \sum_{k : k \neq j} b_{jk}(t) \md N_{jk}(t) \right)
\end{align}
for $t \in (0,T]$ while $A^{\text{b}}(0)$ is some initial deterministic premium. Here $b_j$ are deterministic sojourn payments  and $b_{jk}$ are deterministic transition payments all assumed measurable and bounded on bounded intervals. 
To keep the notation simple, we have disregarded lump-sum payments. An extension to more general payments is straightforward. 

Finally, we specify the structure of taxes and expenses. The tax rate $\gamma\in[0,1)$ is assumed to be constant in time and deterministic. In particular, the tax rate does not depend on the history of the insurance contract(s) or the history of the financial market. The expense rate $\delta$ is assumed to take the form
\begin{align*}
\delta(t)
=
\sum_{j\in\mathcal{J}} \1{Z(t)=j} \delta_j(t)
\end{align*}
for $t \in (0,T]$, where $\delta_j$ are continuous state-wise expense rates assumed deterministic. It follows that the expense rate only depends on the history of the insurance contract(s) through the present state, and that the expense rate does not depend on the history of the financial market.

\subsection{Tax- and expense-modified risk-minimization}\label{subsec:ex_riskmin}

Based on the Markovianity of the short-rate $r$, define $F$ and $F^{1-\gamma}$ by
\begin{align*}
F(t,r(t),s)
&=
\E^Q\!\left[\left.  e^{-\int_t^s r(u) \md u} \, \right| r(t) \right]
=
\E^Q\!\left[\left.  e^{-\int_t^s r(u) \md u} \, \right| \mathcal{F}(t) \right]
\!, \\
F^{1-\gamma}(t,r(t),s)
&=
\E^Q\!\left[\left.  e^{-\int_t^s (1-\gamma)r(u) \md u} \, \right| r(t) \right]
=
\E^Q\!\left[\left.  e^{-\int_t^s (1-\gamma)r(u) \md u} \, \right| \mathcal{F}(t) \right]\!.
\end{align*}
Note that $F(t,r(t),T)=S_1(t)$. Also note that
\begin{align*}
\md (1-\gamma)r(t)
=
\kappa\left((1-\gamma)\theta-(1-\gamma)r(t)\right)\! \md t + (1-\gamma)\sigma \md W(t),
\end{align*}
so that $t \mapsto (1-\gamma)r(t)$ is another Ornstein-Uhlenbeck process. Using explicit results for the Vasicek term structure model, see e.g.\ \cite{bjork2009} Proposition 24.3, we then find that
\begin{align}\label{eq:Vasicek_help}
F^{1-\gamma}_r(t,r(t),s)
=
(1-\gamma)F_r(t,r(t),s)\frac{F^{1-\gamma}(t,r(t),s)}{F(t,r(t),s)},
\end{align}
where $F_r(t,r,s)=\pderiv{}{r}F(t,r,s)$ and similarly for $F^{1-\gamma}_r(t,r,s)$.

Define also so-called expense deflated transition probabilities $p^{-\delta}$ by
\begin{align*}
p_{ij}^{-\delta}(t,s)
=
\E^Q\!\left[\left.  \mathds{1}_{\{Z(s) = j\}} e^{\int_t^s \delta_{Z(u)}(u) \md u} \, \right| Z(t) = i \right]\!.
\end{align*}
If the expense rates are zero, these are in fact the usual transition probabilities. In the general case, it follows from Appendix~\ref{ap:A} that the expense deflated transition probabilities satisfy systems of ordinary differential equations similar to Kolmogorov's backward and forward differential equations.

The tax- and expense-modified version $\tilde{A}^{\text{b}}$ of the discounted insurance payment process is given by
\begin{align*}
	\tilde{A}^{\text{b}}(t) &= A^{\text{b}}(0) +  \int_0^t e^{-\int_0^s \left((1-\gamma)r(u) - \delta(u)\right) \md u} \md A^{\text{b}}(s),
\end{align*}
confer with e.g.\ \eqref{eq:tilde_check_ident}.

The following approach follows along the lines of Section 3 in \cite{thmoller01}. Using the independence between $Z$ and the financial market, it can be shown that the intrinsic value process associated with $\tilde{A}^{\text{b}}$ can be written as
\begin{align}\label{eq:intrinsic_ex}
	\mathcal{V}^{\tilde{A}^{\text{b}}}(t) 
	&= \E^Q\!\left[ \left. \tilde{A}^{\text{b}}(T) \, \right| \mathcal{F}(t)\right]  \\ \label{eq:intrinsic_ex2}
	&= \tilde{A}^{\text{b}}(t) + e^{-\int_0^t \left((1-\gamma)r(u) - \delta(u)\right) \md u}
	\int_t^T F^{1-\gamma}(t,r(t),s) Y_{Z(t)}^{-\delta}(t,s)\md s,
\end{align}
where $Y_i^{-\delta}$ is given by
\begin{align*}
Y_i^{-\delta}(t,s)
=
\sum_{j \in \mathcal{J}} p_{ij}^{-\delta}(t,s)\!\left(
b_j(s) + \sum_{k : k \neq j} \mu_{jk}(s)b_{jk}(s)
\right)\!.
\end{align*}
This expression for the intrinsic value process is comparable to that of \cite{thmoller01} p.\ 426, and we may therefore proceed using the same techniques as in \cite{thmoller01} pp.\ 442--444.

Define so-called state-wise prospective reserves $V^{1-\gamma,\delta}_i$ by
\begin{align*}
V^{1-\gamma,\delta}_i(t)
=
\int_t^T F^{1-\gamma}(t,r(t),s)Y_i^{-\delta}(t,s) \md s.
\end{align*}
We are now ready to state the relevant Galtchouk-Kunita-Watanabe decomposition:
\begin{lemma}\label{lemma:GKW_vasi}
The Galtchouk-Kunita-Watanabe decomposition of $\mathcal{V}^{\tilde{A}^{\text{b}}}$ is given by
\begin{align*}
	\mathcal{V}^{\tilde{A}^{\text{b}}}(t) 
	=
	\mathcal{V}^{\tilde{A}^{\text{b}}}(0) 
	+
	\int_0^t \mathds{1}_{\{Z(s-)=i\}} \xi_i(s) \md S^*_1(s)
	+
	\sum_{j \in \mathcal{J}}\sum_{k : k \neq j} \int_0^t v_{jk}(s) \md M_{jk}(s).
\end{align*}
where
\begin{align}\label{eq:lemma_1}
\xi_i(t)
&=
(1-\gamma)e^{\int_0^t \left(\gamma r(u) + \delta(u)\right) \md u} \int_t^T
\frac{F_r(t,r(t),s)}{F_r(t,r(t),T)}
\frac{F^{1-\gamma}(t,r(t),s)}{F(t,r(t),s)}
Y_i^{-\delta}(t,s) \md s, \\ \label{eq:lemma_2}
v_{jk}(t)
&=
e^{-\int_0^t \left((1-\gamma)r(u) - \delta(u)\right) \md u}
\left(b_{jk}(t) + V^{1-\gamma,\delta}_k(t) - V^{1-\gamma,\delta}_j(t)\right).
\end{align}
\end{lemma}
\begin{proof}
The proof mirrors the proof of~\cite{thmoller01} Lemma 3.2, although under relaxed regularity conditions. We therefore only sketch the essential steps with a focus on the complications that arrive due to the inclusion of taxes and expenses. 

First, one takes a closer look at the dynamics of
\begin{align*}
e^{-\int_0^t \left((1-\gamma)r(u) - \delta(u)\right) \md u} V^{1-\gamma,\delta}_i(t).
\end{align*}
Using a system of ordinary differential equations similar to Kolmogorov's backward differential equations, see Appendix~\ref{ap:A}, and then proceeding along the lines of~\cite{thmoller01} pp.\ 443--444, it is then possible to show that
\begin{align*}
	\mathcal{V}^{\tilde{A}^{\text{b}}}(t) 
	=
	\mathcal{V}^{\tilde{A}^{\text{b}}}(0) 
	+
	\int_0^t \mathds{1}_{\{Z(s-)=i\}} \tilde{\xi}_i(s) \md W(s)
	+
	\sum_{j \in \mathcal{J}}\sum_{k : k \neq j} \int_0^t v_{jk}(s) \md M_{jk}(s),
\end{align*}
where $v_{jk}$ is given by \eqref{eq:lemma_2} and
\begin{align}\label{eq:tildeXi}
\tilde{\xi}_i(t)
&=
e^{-\int_0^t \left((1-\gamma)r(u) - \delta(u)\right) \md u}
\sigma
\int_t^T
F_r^{1-\gamma}(t,r(t),s)
Y_i^{-\delta}(t,s) \md s.
\end{align}
Using the Vasicek term structures, we next find that
\begin{align*}
\tilde{\xi}_i(t)
&=
e^{-\int_0^t \left((1-\gamma)r(u) - \delta(u)\right) \md u} (1-\gamma)
\sigma
\int_t^T
F_r(t,r(t),s)\frac{F^{1-\gamma}(t,r(t),s)}{F(t,r(t),s)}
Y_i^{-\delta}(t,s) \md s \\
&=
e^{\int_0^t \left(\gamma r(u) + \delta(u)\right) \md u} (1-\gamma)
\int_t^T
\frac{F_r(t,r(t),s)}{F_r(t,r(t),T)}\frac{F^{1-\gamma}(t,r(t),s)}{F(t,r(t),s)}
Y_i^{-\delta}(t,s) \md s \\
&\times e^{-\int_0^t r(s) \md s}
F_r(t,r(t),T)
\, \sigma,
\end{align*}
see also~\eqref{eq:Vasicek_help}. In other words,
\begin{align*}
\tilde{\xi}_i(t)
=
\xi_i(t)
e^{-\int_0^t r(s) \md s}
F_r(t,r(t),T)
\, \sigma
\end{align*}
with $\xi_i$ defined by \eqref{eq:lemma_1}. Now recall that
\begin{align*}
\md S^*_1(t)
=
e^{-\int_0^t r(s) \md s}
F_r(t,r(t),T)
\, \sigma \md W(t),
\end{align*}
from which it follows that
\begin{align*}
\tilde{\xi}_i(t) \md W(t)
=
\xi_i(t) \md S^*_1(t),
\end{align*}
completing the sketch of proof.
\end{proof}
\begin{remark}\label{rmk:extend}
For a general short rate model, the proof technique of Lemma~\ref{lemma:GKW_vasi} still applies and similar results as in the Vasicek term structure model remain obtainable. Assume the short rate satisfies the stochastic differential equation
\begin{align*}
\md r(t)
=
\alpha(t,r(t)) \md t + \sigma(t,r(t)) \md W(t),
\end{align*}
with $W$ still a standard Brownian motion under $Q$ and where $\alpha$ and $\sigma$ are functions satisfying certain Lipschitz conditions. Imposing suitable additional regularity conditions on the short rate model (equivalently, the term structure model), one finds
\begin{align*}
\tilde{\xi}_i(t)
&=
e^{-\int_0^t \left((1-\gamma)r(u) - \delta(u)\right) \md u}
\sigma(t,r(t))
\int_t^T
F^{1-\gamma}_r(t,r(t),s)
Y_i^{-\delta}(t,s) \md s \\
&=
e^{\int_0^t \left(\gamma r(u) + \delta(u)\right) \md u}
\int_t^T
\frac{F^{1-\gamma}_r(t,r(t),s)}{F_r(t,r(t),T)}
Y_i^{-\delta}(t,s) \md s \, \,
e^{-\int_0^t r(s) \md s}
F_r(t,r(t),T)
\, \sigma(t,r(t)),
\end{align*}
by following along the lines of the sketch of proof of Lemma~\ref{lemma:GKW_vasi}. This results in the following Galtchouk-Kunita-Watanabe decomposition:
\begin{align*}
	\mathcal{V}^{\tilde{A}^{\text{b}}}(t) 
	=
	\mathcal{V}^{\tilde{A}^{\text{b}}}(0) 
	+
	\int_0^t \mathds{1}_{\{Z(s-)=i\}} \xi_i(s) \md S^*_1(s)
	+
	\sum_{j \in \mathcal{J}}\sum_{k : k \neq j} \int_0^t v_{jk}(s) \md M_{jk}(s).
\end{align*}
where now
\begin{align*}
\xi_i(t)
&=
e^{\int_0^t \left(\gamma r(u) + \delta(u)\right) \md u} \int_t^T
\frac{F^{1-\gamma}_r(t,r(t),s)}{F_r(t,r(t),T)}
Y_i^{-\delta}(t,s) \md s, \\
v_{jk}(t)
&=
e^{-\int_0^t \left((1-\gamma)r(u) - \delta(u)\right) \md u}
\left(b_{jk}(t) + V^{1-\gamma,\delta}_k(t) - V^{1-\gamma,\delta}_j(t)\right).
\end{align*}
\end{remark}
As we have identified the Galtchouk-Kunita-Watanabe decomposition of $\mathcal{V}^{\tilde{A}^{\text{b}}}$, we are now ready to apply the results on tax- and expense-modified risk-minimization to obtain the main result of this section:
\begin{theorem}\label{thm:example}
The unique risk-minimizing investment strategy $\tilde{h}$ in the setting of Section~\ref{sec:ex} is given as follows:
\begin{align*}
\tilde{h}_1(t)
&=
\int_t^T
\frac{F_r(t,r(t),s)}{F_r(t,r(t),T)}\frac{F^{1-\gamma}(t,r(t),s)
}
{F(t,r(t),s)
} Y_{Z(t-)}^{-\delta}(t,s) \md s, \\
\tilde{h}_0(t)
&=
S_0^{-1}(t)\left(V^{1-\gamma,\delta}_{Z(t)}(t) - \tilde{h}_1(t)S_1(t)\right)\!.
\end{align*}
The associated value process is
\begin{align*}
V(\tilde{h},t)
=
V^{1-\gamma,\delta}_{Z(t)}(t).
\end{align*}
\end{theorem}
\begin{proof}
The first statement follows immediately by combining Lemma~\ref{lemma:GKW_vasi} with Theorem~\ref{te:rmtax} and the observation
\begin{align*}
e^{-\int_0^t \left(\gamma r(u) + \delta(u)\right) \md u}
\left(
\mathcal{V}^{\tilde{A}^{\text{b}}}(t) -\tilde{A}^{\text{b}}(t)
\right)
&=
S_0^{-1}(t)
\int_t^T F^{1-\gamma}(t,r(t),s) Y_{Z(t)}^{-\delta}(t,s)\md s \\
&=
S_0^{-1}(t) V^{1-\gamma,\delta}_{Z(t)}(t),
\end{align*}
confer with~\eqref{eq:intrinsic_ex2}. The last statement follows by direct calculations from the first statement and \eqref{eq:Vhtdef}.
\end{proof}

\begin{remark}\label{rmk:extend2}
In Remark~\ref{rmk:extend} we discussed extensions of the Galtchouk-Kunita-Watanabe decomposition of Lemma~\ref{lemma:GKW_vasi} to general short rate models. Based on this discussion, we can extend the conclusions of Theorem~\ref{thm:example} to the general framework of Remark~\ref{rmk:extend} in the following manner.

Assume the short rate satisfies the stochastic differential equation
\begin{align*}
\md r(t)
=
\alpha(t,r(t)) \md t + \sigma(t,r(t)) \md W(t),
\end{align*}
with $W$ still a standard Brownian motion under $Q$ and where $\alpha$ and $\sigma$ are functions satisfying certain Lipschitz conditions. Imposing suitable additional regularity conditions, the unique risk-minimizing investment strategy $\tilde{h}$ is given by
\begin{align*}
\tilde{h}_1(t)
&=
\int_t^T \frac{1}{1-\gamma}
\frac{F^{1-\gamma}_r(t,r(t),s)
}
{F_r(t,r(t),T)
} Y_{Z(t-)}^{-\delta}(t,s) \md s, \\
\tilde{h}_0(t)
&=
S_0^{-1}(t)\left(V^{1-\gamma,\delta}_{Z(t)}(t) - \tilde{h}_1(t)S_1(t)\right)\!.
\end{align*}
The associated value process is still
\begin{align*}
V(\tilde{h},t)
=
V^{1-\gamma,\delta}_{Z(t)}(t).
\end{align*}
\end{remark}

\subsection{Discussion}\label{subsec:ex_dis}

The unique risk-minimizing investment strategy $\tilde{h}$ given by Theorem~\ref{thm:example} is a modification of the classic strategy without taxes and expenses. The quantity $Y_{Z(t-)}^{-\delta}(t,s)$ is the expected future (diversified) rate of payments at time $s$ given the present state of the insurance contract(s) at time $t$ while taking future state-wise expenses into account; it can be interpreted as an expected expense-modified cash flow. 

To cover the future infinitesimal expected payment at time $s$ the strategy dictates an investment of
\begin{align}\label{eq:product}
\frac{F_r(t,r(t),s)}{F_r(t,r(t),T)}\frac{F^{1-\gamma}(t,r(t),s)
}
{F(t,r(t),s)
} Y_{Z(t-)}^{-\delta}(t,s) \md s
\end{align}
into the bond. Thus what regards investment in the risky asset, taxes are taken into account by increasing the investment by a factor of
\begin{align*}
\frac{F^{1-\gamma}(t,r(t),s)
}
{F(t,r(t),s)
} \geq 1,
\end{align*}
confer also with the discussion in~\cite{buchardtmoeller2018} Section 4, in particular~\cite{buchardtmoeller2018} Subsection 4.3.2. The product structure of \eqref{eq:product} w.r.t.\ taxes and expenses is a direct consequence of the independence between market and insurance risks and the fact that tax rate does not depend on the history of the insurance contract(s) nor the market while the expense rate only depends on the history of the insurance contract(s).

To explicitly compute the risk-minimizing investment strategy and the associated value process, one needs to calculate $F$, $F_r$, and $F^{1-\gamma}$ as well as the expense deflated transition probabilities $p_{ij}^{-\delta}$. For the Vasicek term structure model, where in particular $t \mapsto (1-\gamma)r$ is another Ornstein-Uhlenbeck process, the former quantities have closed-form expressions and are therefore easily calculated. Because the expense deflated transition probabilities $p_{ij}^{-\delta}$ can be found by solving a system of ordinary differential equations similar to Kolmogorov's forward differential equations, see Appendix~\ref{ap:A}, this establishes a simple scheme for the computation of the risk-minimizing investment strategy and the associated value process.

In Remark~\ref{rmk:extend2} we elaborated on how to extend Theorem~\ref{thm:example} to general short rate models. If the model is affine, the relevant quantities needed for computation of the risk-minimizing investment strategy and the associated value process, i.e.\ $F$, $F_r^{1-\gamma}$, and $F^{1-\gamma}$, can be calculated by solving systems of ordinary differential equations, see  \cite{duffiepansingleton,buchardt2016}. The model of Section~\ref{sec:ex} can therefore easily be implemented in practice; and furthermore, the extension to general affine term structure models is relatively straightforward.

\section*{Acknowledgments and declarations of interest}

Christian Furrer’s research is partly funded by the Innovation Fund Denmark (IFD) under File No.\ 7038-00007. The authors declare no competing interests.
  
\appendix

\section{Deflated transition probabilities}\label{ap:A}

Let $(\Omega,\mathcal{F},P)$ be a background probability space, and let $Z$ be a Markovian jump process with values in a finite set $\mathcal{J}=\{0,1,\ldots,n-1\}$. Let $N$ be the multivariate counting process associated with $Z$. The transition probabilities of $Z$ are given by $n\times n$-matrices $p(t,s)$ for $0 \leq t \leq s < \infty$, where
\begin{align*}
p_{ij}(t,s)
=
P\!\left[\left.  Z(s) = j \, \right| Z(t) = i \right]\!,
\end{align*}
and the transition probabilities satisfy the Chapman-Kolmogorov equation.

We assume the existence of continuous transition intensities $\mu_{jk}$, when each counting process $N_{jk}$ has intensity process $\lambda_{jk}$ given by
\begin{align*}
\lambda_{jk}(t) = \mathds{1}_{\{Z(t-) = j\}} \mu_{jk}(t).
\end{align*}
We can then take regular versions of the conditional distributions for which the transition probabilities $p$ satisfy
\begin{align*}
\mu(t)
=
\lim_{h \searrow 0} \frac{
p(t,t+h) - p(t,t)
}
{h},
\end{align*}
where $\mu$ are $n\times n$-matrices with diagonal elements
\begin{align*}
\mu_{jj} = -\sum_{k : k \neq j} \mu_{jk}.
\end{align*}
Furthermore, the transition probabilities satisfy Kolmogorov' backward and forward differential equations.

Let $\delta$ be $n\times1$-dimensional with deterministic and continuous elements $t \mapsto \delta_i(t)$. Quantities of interest are (corresponding regular versions) of
\begin{align*}
 p^\delta_{ij}(t,s) &= \E\left[\left. 1_{\{Z(s)=j\}} e^{-\int_t^s \delta_{Z(u)}(u) \md u} \right| Z(t) = i\right]\!,
\end{align*}
which we term \textit{$\delta$-deflated transition probabilities}. When $\delta\equiv 0_{n\times 1}$, we see that these quantities are in fact the transition probabilities. When $\delta\equiv 1_{n\times 1} f$ for some deterministic and continuous function $t \mapsto f(t)$, we see that
\begin{align*}
p^\delta_{ij}(t,s)
=
e^{-\int_t^s f(u) \md u}p_{ij}(t,s).
\end{align*}
The $\delta$-deflated transition probabilities satisfy systems of ordinary differential equations similar to Kolmogorov's backward and forward differential equations:
\begin{lemma}\label{lemma:ap}
The $\delta$-deflated transition probabilities satisfy the forward ordinary differential equation system
\begin{align*}
    \pderiv{}{s}p^\delta(t,s) = p^\delta(t,s) \left[\mu - \emph{diag}(\delta)\right]\!(s),
\end{align*}
and the backward ordinary differential equation system
\begin{align*}
    \pderiv{}{t}p^\delta(t,s) = - \left[\mu - \emph{diag}(\delta)\right]\!(t)p^\delta(t,s),
\end{align*}
with boundary conditions $p^\delta(t,t) = \emph{diag}(1_{n \times 1})$.
\end{lemma}

\begin{proof}
The boundary conditions are evident. We first prove the forward differential equations. Define the $1 \times n$-dimensional indicator process $I$ by
\begin{align*}
     I_i(t) = 1_{\{Z(t)=i\}}.
\end{align*}
For fixed $t_0\geq 0$, define also the $1 \times n$-dimensional process $X$ by
\begin{align*}
    X(t) = I(t) e^{-\int_{t_0}^t \delta_{Z(u)}(u) \md u}.
\end{align*}
View $N$ as $n\times n$-matrices with diagonal elements
\begin{align*}
N_{jj} = - \sum_{k : k \neq j} N_{jk}
\end{align*}
In similar fashion, view $\lambda$ as $n\times n$-matrices with diagonal elements
\begin{align*}
\lambda_{jj} = - \sum_{k : k \neq j} \lambda_{jk},
\end{align*}
such that $\lambda_{jj}(t) = I_j(t-) \mu_{jj}(t)$.

Recalling that $\mathrm{d}I_i=\sum_{j \neq i} (\mathrm{d}N_{ji} - \mathrm{d}N_{ij}) = \sum_{j} \mathrm{d}N_{ji}$, we see that
\begin{align*}
\mathrm{d}I
=
1_{1 \times n}\mathrm{d}N.
\end{align*}
Because the compensated jump processes
\begin{align*}
t \mapsto N_{ij}(t) - \int_0^t \lambda_{ij}(s) \md s
\end{align*}
are martingales, we find that
\begin{align}\label{eq:apAI}
\mathrm{d}I(t)
&=
\mathrm{d}M(t)
+
1_{1 \times n}\lambda(t) \mathrm{d}t \nonumber \\
&=
\mathrm{d}M(t)
+
I(t-)\mu(t) \mathrm{d}t \nonumber \\
&=
\mathrm{d}M(t)
+
I(t)\mu(t) \mathrm{d}t,
\end{align}
where $M$ is a $1 \times n$-dimensional martingale given by
\begin{align*}
\mathrm{d}M(t)
=
1_{1 \times n}\left(\mathrm{d}N(t) - \lambda(t) \mathrm{d}t\right).
\end{align*} 
Integration by parts now yields 
\begin{align*}
\mathrm{d}X(t)
&=
\left(\mathrm{d}I(t)\right) e^{-\int_{t_0}^t \delta_{Z(u)}(u) \md u}
-
I(t) \delta_{Z(t)}(t) e^{-\int_{t_0}^t \delta_{Z(u)}(u) \md u} \mathrm{d}t \\
&=
X(t)\mu(t)\mathrm{d}t-X(t)\delta_{Z(t)}(t)\mathrm{d}t 
+e^{-\int_{t_0}^t \delta_{Z(u)}(u) \md u} \mathrm{d}M(t).
\end{align*}
By definition of $X$,
\begin{align*}
\E \left[ \left. X_j(t) \right| Z(t_0) = i\right]
&=
p_{ij}^\delta(t_0,t), \\
\E \left[ \left. \delta_{Z(t)}(t)X_j(t) \right| Z(t_0) = i\right]
&=
\delta_j(t)p_{ij}^\delta(t_0,t).
\end{align*}
The latter corresponds to the $(i,j)$'th element of the matrix product of $p^\delta(t_0,t)$ and the diagonal matrix with diagonal $\delta(t)$. Collecting all terms, it then follows from Fubini's theorem and the martingale properties of $M$ that
\begin{align*}
p^\delta(t_0,t)
=
p^\delta(t_0,t_0)
+
\int_{t_0}^t p^\delta(t_0,s) \left[\mu-\text{diag}(\delta)\right]\!(s) \md s.
\end{align*}
The forward differential equations now follow by differentiation w.r.t.\ $t$.

We now turn our attention to the backward differential equations.
For fixed $s \geq 0$ define the $1 \times n$-dimensional martingale $Y$ by
\begin{align*}
Y(t) &= \E \left[ \left. I(s) e^{-\int_0^s \delta_{Z(u)}(u)\md u} \right| \mathcal{F}(t) \right] \\
        &=  I(t) e^{-\int_0^t \delta_{Z(u)}(u) \md u} p^{\delta}(t,s),
\end{align*} 
where $0 \leq t \leq s$.

Integration by parts now yields
\begin{align*}
e^{\int_0^t \delta_{Z(u)}(u) \md u} \mathrm{d}Y(t)
=& \,
\mathrm{d}\!\left(I(t)p^{\delta}(t,s)\right)
-
I(t) \delta_{Z(t)}(t) p^{\delta}(t,s) \mathrm{d}t\\
=& \,
I(t)p^{\delta}(\mathrm{d}t,s)
+
I(t)\mu(t)p^{\delta}(t,s) \mathrm{d}t
-
I(t)\delta_{Z(t)}(t)p^{\delta}(t,s) \mathrm{d}t\\
&+
\mathrm{d}M(t) p^{\delta}(t,s),
\end{align*}
where we have used~\eqref{eq:apAI}. Because $Y$ and $M$ are $1\times n$-dimensional martingales, we find using martingale representation theory that $p^{\delta}(t,s)$ is differentiable in $t$ and that
\begin{align*}
I(t)\pderiv{}{t}p^{\delta}(t,s)
=
-
\left(
I(t)\mu(t)p^{\delta}(t,s)
-
I(t)\delta_{Z(t)}(t)p^{\delta}(t,s)
\right)\!.
\end{align*}
The backward differential equations can now be established by taking a closer look at this expression on each event $\{Z(t)=i\}$ for varying $i$. For example, on $\{Z(t)=i\}$,
\begin{align*}
\left[I(t)\delta_{Z(t)}(t)p^{\delta}(t,s)\right]_j
&=
\delta_i(t)p^\delta_{ij}(t,s),
\end{align*}
which corresponds to the $(i,j)$'th element of the matrix product between the diagonal matrix with diagonal $\delta(t)$ and the matrix $p^\delta(t,s)$. By additional observations of the same kind, the backward differential equations follow. This completes the proof.
\end{proof} 

\bibliography{../refsSmall}{}
\bibliographystyle{plainnat}

\end{document}